\documentclass[12pt]{article}

\usepackage{bm}
\usepackage{subfigure}
\usepackage{graphics}
\usepackage{epsfig}
\usepackage{amsfonts}
\usepackage{amsmath}
\usepackage{dsfont}
\usepackage{pifont}
\usepackage{bbm}
\usepackage{multirow}
\usepackage{latexsym}
\usepackage{verbatim}
\usepackage{arydshln}
\usepackage{cancel}
\usepackage{slashed}
\usepackage{natbib}
\usepackage{url}
\usepackage{mathrsfs}
\usepackage{amsmath,epsfig,amssymb,amsfonts,amsthm,epsfig,verbatim}

\newcommand{\indicator}[1]{\mathbbm{1}_{\left[ {#1} \right] }}
\def\bbbr{{\rm I\!R}}
\def\bbbn{{\rm I\!N}}
\def\Skip{\par\bigskip\nobreak}

\def\bbbr{{\rm I\!R}}
\def\bbbn{{\rm I\!N}}
\def\Skip{\par\bigskip\nobreak}
\newtheorem{example}{Example}
\newtheorem{lemma}{Lemma}
\newtheorem{corollary}{Corollary}
\newtheorem{theorem}{Theorem}
\newtheorem{definition}{Definition}
\newtheorem{proposition}{Proposition}
\def\Skip{\par\bigskip\nobreak}
\def\dj{d\kern-0.4em\char"16\kern-0.1em}
\def\Dj{\mbox{\raise0.3ex\hbox{-}\kern-0.4em D}}

\DeclareMathAlphabet{\mathpzc}{OT1}{pzc}{m}{it}

\begin{document}
\pagestyle{plain}

\makeatletter
\@addtoreset{equation}{section}
\makeatother
\renewcommand{\thesection}{\arabic{section}}
\renewcommand{\theequation}{\thesection.\arabic{equation}}
\renewcommand{\thefootnote}{\arabic{footnote}}

\setcounter{page}{1}
\setcounter{footnote}{0}

\begin{titlepage}
\begin{flushright}
\small ~~
\end{flushright}

\bigskip

\begin{center}

\vskip 0cm

{\LARGE \bf   {Level Sets Based Distances for Probability Measures and Ensembles with Applications}} \\[6mm]

\vskip 0.5cm

{\bf  Alberto Mu\~noz$^1$,
Gabriel Martos$^1$ and Javier Gonz\'alez$^2$}\\

\vskip 25pt


{\em $^1$Department of Statistics, University Carlos III of Madrid \\ Spain.  C/ Madrid, 126 - 28903, Getafe (Madrid), Spain.}\\
{\small {\tt \ {alberto.munoz@uc3m.es, gabriel.martos@uc3m.es }}}

\Skip
{\em
             $^2$Sheffield Institute for Translational Neuroscience, \\Department of Computer Science, University of Sheffield.\\
               Glossop Road S10 2HQ, Sheffield, UK.\\
}
{\small {\tt \ { j.h.gonzalez@sheffield.ac.uk }}}

\vskip 0.8cm

\end{center}

\vskip 1cm

\begin{center}

{\bf ABSTRACT}\\[3ex]

\begin{minipage}{13cm}
\small
In this paper we study Probability Measures (PM) from a functional point of view: we show that PMs can be considered as
functionals (generalized functions) that belong to some functional space endowed with an inner product.
This approach allows us to introduce a new family of distances for PMs, based on the action of the PM functionals on `interesting' functions of the sample. We propose a specific (non parametric) metric for PMs belonging to this class, based on the estimation of density level sets. Some real and simulated data sets are used to measure the performance of the proposed distance against a battery of distances widely used in Statistics and related areas.

\vspace{0.5cm}

\end{minipage}

\end{center}

\vfill

\end{titlepage}

\section{Introduction } \label{sec:intro}
Probability metrics, also known as statistical distances, are of fundamental importance in Statistics. In essence, a probability metric it is a measure that quantifies how (dis)similar are two random quantities, in particular two probability measures (PM). Typical examples of the use of probability metrics in Statistics are homogeneity, independence and goodness of fit tests. For instance there are some goodness of fit tests based on the use of the $\chi^2$ distance and others that use the Kolmogorov-Smirnoff statistics, which corresponds to the choice of the supremum distance between two PMs. There exist a large literature about probability metrics, for a summary review on interesting probability metrics and theoretical results refer to \citep{bib:deza,AMuller,VZolota} and references therein.
 
Statistical distances are also extensively used in several applications related to Machine Learning and Pattern Recognition. Several examples can be found, for instance, in Clustering \citep{nielsen2011burbea,bib:berenjae}, Image Analysis \citep{levina2001earth,rubner2000earth}, Bioinfomatics \citep{minas2013distance,saez2013comparative}, Time Series Analysis \citep{ryabko2012reducing,bib:moon} or Text Mining \citep{bib:lebanon}, just to name a few.

In practical situations we do not know explicitly the underlying distribution of the data at hand, and we need to compute a distance between probability measures by using a finite data sample. In this context, the computation of a distance between PMs that rely on the use of non-parametric density estimations often is computationally difficult and the rate of convergence of the estimated distance is usually slow \citep{bib:nguy1,wang2005divergence,stone1980optimal}. In this work we extend the preliminary idea presented in \citep{munoz2012new}, that consist in considering PMs as points in a functional space endowed with an inner product. We derive then different distances for PMs from the metric structure inherited from the ambient inner product. We propose particular instances of such metrics for PMs based on the estimation of density level sets regions avoiding in this way the difficult task of density estimation.

This article is organized as follows: In Section \ref{sec:distances} we review some distances for PMs and represent probability measures as generalized functions; next we define general distances acting on the Schwartz distribution space that contains the PMs. Section \ref{sec:metriconls} presents a new distance built according to this point of view. Section \ref{sec:exp} illustrates the theory with some simulated and real data sets. Section \ref{sec:concl} concludes.

\section{Distances for probability distributions}	
\label{sec:distances}
Several well known statistical distances and divergence measures are special cases of $f$-divergences \citep{bib:csisz}. 
Consider two PMs, say $\mathbb{P}$ and $\mathbb{Q}$, defined on a measurable space $(X, \mathcal{F},\mu)$, where $X$ is a sample space, $\mathcal{F}$ a $\sigma$-algebra of measurable subsets of $X$ and $\mu:\mathcal{F}\rightarrow \bbbr^+$ the Lebesgue measure. For a convex function $f$ and assuming that $\mathbb{P}$ is absolutely continuous with respect to $\mathbb{Q}$, then the $f$-divergence from
$\mathbb{P}$ to $\mathbb{Q}$ is defined by:
\begin{equation}\label{eq:fdiv}
d_{f}(\mathbb{P},\mathbb{Q})=\int\limits_{X}f \left( \frac{d\mathbb{P}}{d\mathbb{Q}}\right) d\mathbb{Q}.
\end{equation}
Some well known particular cases: for $f(t)=\frac{|t-1|}{2}$ we obtain the {\it Total Variation} metric; $f(t)=(t-1)^2$ yields the { \it $\chi^2$}-distance; $f(t)=(\sqrt{t}-1)^2$ yields the {\it Hellinger} distance.

The second important family of dissimilarities between probability distributions is made up of Bregman Divergences: Consider a continuously-differentiable real-valued and strictly convex function $\varphi$ and define:
\begin{equation}
d_{\varphi}(\mathbb{P},\mathbb{Q})=\int\limits_{X}\left(\varphi(p)-\varphi(q)-(p-q)\varphi'(q)\right)d\mu(x),
\end{equation}

where $p$ and $q$ represent the density functions for $\mathbb{P}$ and $\mathbb{Q}$ respectively and $\varphi'(q)$ is the derivative of $\varphi$ evaluated at $q$ (see \citep{bib:friyik, bib:ciko} for further details). Some examples of Bregman divergences: $\varphi(t)=t^2$ $d_{\varphi}(\mathbb{P},\mathbb{Q})$ yields the Euclidean distance between $p$ and $q$  (in $L_2$); $\varphi(t)=t\log(t)$ yields the {\it Kullback Leibler} (KL) Divergence; and for $\varphi(t)= -\log(t)$ we obtain the {\it Itakura}-{\it Saito} distance.
In general $d_{f}$ and $d_{\varphi}$ are not metrics because the lack of symmetry and because they do not necessarily satisfy the triangle inequality. 

A third interesting family of PM distances are integral probability metrics (IPM) \citep{VZolota,AMuller}. Consider a class of real-valued bounded measurable functions on $X$, say $\mathcal H$, and define the IPM between $\mathbb{P}$ and $\mathbb{Q}$ as
\begin{equation}\label{eq:ipmm}
d_{\mathcal{H}}(\mathbb{P},\mathbb{Q})=\sup_{f \in \mathcal{H}}\left| \int f d\mathbb{P} - \int f d\mathbb{Q}\right|.
\end{equation}

If we choose $\mathcal{H}$ as the space of bounded functions such that $h \in \mathcal{H}$ if $\|h\|_{\infty}\leq1$, then $d_{\mathcal{H}}$ is the Total Variation metric; 
when $\mathcal{H}=\{\prod_{i=1}^d\indicator{(-\infty,x_i)}: x=(x_1,\dots,x_d) \in \mathbb{R}^d\}$, $d_{\mathcal{H}}$ is the Kolmogorov distance; if $\mathcal{H}=\{e^{\sqrt{-1}\langle \omega,\; \cdot \;\rangle}: \omega \in \mathbb{R}^d\}$ the metric computes the maximum difference between characteristics functions. In \citep{Sriperumb_HSEmbsd} the authors propose to choose $\mathcal{H}$ as a Reproducing Kernel Hilbert Space and study conditions on $\mathcal{H}$ to obtain proper metrics $d_{\mathcal{H}}$.

In practice, the obvious problem to implement the above described distance functions is that we do not know the density (or distribution) functions corresponding to the samples under consideration. For instance suppose we
want to estimate the KL divergence (a particular case of Eq. \eqref{eq:fdiv} taking $f(t) = -\log t$) between two continuous distributions $\mathbb{P}$ and $\mathbb{Q}$ from two given samples. In order to do this we must choose a number of regions, $N$, and then estimate the density functions for $\mathbb{P}$ and $\mathbb{Q}$ in the $N$ regions to yield the following estimation:
\begin{equation}
\widehat{KL}(\mathbb{P},\mathbb{Q})=\sum\limits_{i=1}^{N}\hat{p_i}\log\frac{\hat{p_i}}{\hat{q_i}},
\end{equation}
see further details in \citep{bib:boltz}.

As it is well known, the estimation of general distribution functions becomes intractable as dimension arises. This motivates the need of  metrics for probability distributions that do not explicitly rely on the estimation of the corresponding probability/distribution functions. For further details on the sample versions of the above described distance functions and their computational subtleties see \citep{bib:Scott,bib:cha,wang2005divergence,bib:nguy1,Sriperumb_NonparEstim,bib:gorian, bib:rizzo}
and references therein.

To avoid the problem of explicit density function calculations we will adopt the perspective of the generalized function theory of Schwartz (see \citep{bib:zemanian}, for instance), where a function is not specified by its values but by its behavior as a functional on some space of testing functions.

\subsection{Probability measures as Schwartz distributions } \label{sec:distributions}
Consider a measure space $(X, \mathcal{F},\mu)$, where $X$ is a sample space, here a compact set\footnote{A not restrictive assumption in real scenarios, see for instance \citep{moguerza2006support}.} of a real vector space: $X \subset \mathbb{R}^d$, $\mathcal{F}$ a $\sigma$-algebra of measurable subsets of $X$ and $\mu:\mathcal{F}\rightarrow \bbbr^+$ the ambient $\sigma$-additive measure (here the Lebesgue measure). A probability measure $\mathbb{P}$ is a $\sigma$-additive finite measure absolutely continuous w.r.t. $\mu$ that satisfies the three Kolmogorov axioms. By Radon-Nikodym theorem, there exists a measurable function $f:X \rightarrow \bbbr^+$ (the density function) such that $\mathbb{P}(A) = \int_A f d\mu$, and $f=\frac{d\mathbb{P}}{d\mu}$ is the Radon-Nikodym derivative.

A PM can be regarded as a Schwartz distribution (a generalized function, see \citep{bib:strichartz} for an introduction to Distribution Theory):
We consider a vector space $\mathcal{D}$ of test functions. The usual choice for $\mathcal{D}$ is the subset of $C^{\infty}(X)$ made up of functions with compact support.  A distribution (also named generalized function) is a continuous linear functional on  $\mathcal{D}$. A probability measure can be regarded as a Schwartz distribution $\mathbb{P}:\mathcal{D} \rightarrow \bbbr$ by defining $\mathbb{P}(\phi) = \left<\mathbb{P},\phi\right> = \int \phi d\mathbb{P} = \int \phi(x) f(x) d\mu(x) = \left<\phi,f\right>$. When the density function $f\in \mathcal{D}$, then $f$ acts as the representer in the Riesz representation theorem:
 $\mathbb{P}(\cdot) =  \left<\cdot,f\right>$.

In particular, the familiar condition $\mathbb{P}(X)=1$ is equivalent to $\left<\mathbb{P},\indicator{X}\right> = 1$, where the function $\indicator{X}$ belongs to $\mathcal{D}$, being $X$ compact. Note that we do not need to impose that $f \in \mathcal{D}$; only the integral $\left<\phi,f\right>$ should be properly defined for every $\phi \in \mathcal{D}$.

Hence a probability measure/distribution is a continuous linear functional acting on a given function space. Two given linear functionals $\mathbb{P}_1$ and $\mathbb{P}_2$ will be identical (similar) if they act identically (similarly) on every $\phi \in \mathcal{D}$. For instance, if we choose $\phi = Id$, $\mathbb{P}_1(\phi) = \left<f_{\mathbb{P}_1},x\right> =  \int x d\mathbb{P} = \mu_{\mathbb{P}_1}$ and if $\mathbb{P}_1$ and $\mathbb{P}_2$ are `similar' then $\mu_{\mathbb{P}_1} \simeq \mu_{\mathbb{P}_2}$ because $\mathbb{P}_1$ and $\mathbb{P}_2$ are continuous functionals. Similar arguments apply for variance (take $\phi(x) = (x-\mu)^2$) and in general for higher order moments.
For $\phi_{\xi}(x) = e^{ix\xi}$, $\xi \in \bbbr$, we obtain the Fourier transform of the probability measure (called characteristic functions in Statistics), given by $\hat P(\xi) = \left<\mathbb{P},e^{ix\xi} \right> = \int e^{ix\xi} d\mathbb{P}$.

Thus, two PMs can be identified with their action as functionals on the test functions if the set of test functions $\mathcal{D}$ is rich enough and hence, distances between two distributions can be defined from the differences between functional evaluations for appropriately chosen test functions.

\begin{definition}\textbf{(Identification of PM's).} 
Let $\mathcal{D}$ be a set of test functions and $\mathbb{P}$ and $\mathbb{Q}$ two PM's defined on the measure space $(X, \mathcal{F},\mu)$, then we say that $\mathbb{P}=\mathbb{Q}$ on $\mathcal{D}$ if:
$$\left<\mathbb{P},\phi \right> = \left<\mathbb{Q},\phi \right> \quad \forall \phi \in \mathcal{D}.$$
\end{definition}

The key point in our approach is that if we appropriately choose a finite subset of test functions $\{\phi_i\}$, we can compute the distance between the probability measures by calculating a finite number of functional evaluations. 
In the next section we demonstrate that when $\mathcal{D}$ is composed by indicator functions that indicates the regions where the density remains constant, then the set $\mathcal{D}$ is rich enough to identify PM. In the next section we define a distance based on the use of this set of indicator functions. 

\section{A metric based on the estimation of level sets}
\label{sec:metriconls}
We choose $\mathcal{D}$ as $C_c(X)$, the space of all compactly supported, piecewise continuous functions on $X$ (compact), as test functions (remember that $C_c(X)$ is dense in $L_p$). Given two PMs $\mathbb{P}$ and $\mathbb{Q}$, we consider a family of test functions $\{\phi_{i}\}_{i \in I} \subseteq \mathcal{D}$ and then define distances between $\mathbb{P}$ and $\mathbb{Q}$ by weighting terms of the type $d\left(\left<\mathbb{P},\phi_{i} \right>, \left<\mathbb{Q},\phi_{i} \right>\right)$ for ${i \in I}$, where $d$ is some distance function. Our test functions will be indicator functions of $\alpha$-level sets, described below.

Given a PM $\mathbb{P}$ with density function $f_\mathbb{P}$, minimum volume sets are defined by $S_{\alpha}(f_\mathbb{P}) = \{ x \in X|\, f_\mathbb{P}(x) \ge \alpha\}$, such that $P(S_{\alpha}(f_\mathbb{P})) = 1-\nu\,$, where $0 < \nu < 1$. 
If we consider an ordered sequence $0\leq \alpha_1 < \ldots < \alpha_m$, then $S_{\alpha_{i+1}}(f_\mathbb{P}) \subseteq S_{\alpha_{i}}(f_\mathbb{P})$. Let us define the $\alpha_i$-level set: $A_i(\mathbb{P}) = S_{\alpha_i}(f_\mathbb{P}) - S_{\alpha_{i+1}}(f_\mathbb{P})$, $i \in\{1,\ldots,m-1\}$. 
We can choose $\alpha_1\simeq 0$ and $\alpha_m \ge \max_{x\in X} f_\mathbb{P}(x)$ (which exists, given that $X$ is compact and $f_\mathbb{P}$
piecewise continuous); then $\bigcup_i A_i(\mathbb{P}) \simeq \operatorname{Supp}(\mathbb{P}) = \{x \in X |\, f_\mathbb{P}(x) \ne 0\}$ (equality takes place when $m\rightarrow \infty$, $\alpha_1\rightarrow 0$ and $\alpha_m\rightarrow \max_{x\in X} f_\mathbb{P}(x)$).  Given the definition of the $A_i$, if $A_i(\mathbb{P}) = A_i(\mathbb{Q})$ for every $i$ when  $m\rightarrow\infty$, then $\mathbb{P} = \mathbb{Q}$. We formally prove this proposition with the aid of the following theorem.

\begin{definition}{\textbf{($\boldsymbol \alpha_\mathbb{P}^m$ sequence).}}
Given a PM $\mathbb{P}$ defined on the measure space $(X, \mathcal{F},\mu)$, with density function $f_\mathbb{P}$ and $m \in  \mathbb{N}$, define $\boldsymbol \alpha_\mathbb{P}^m=\{\alpha_1 , \dots, \alpha_m \}$ where $0=\alpha_1 < \ldots < \alpha_m = \max_xf_\mathbb{P}(x)$.
\end{definition}

\begin{theorem}{(\textbf{$\alpha$-level set representation of a PM).}}\label{th1}
Given a PM $\mathbb{P}$ defined on the measure space $(X, \mathcal{F},\mu)$, with density function $f_\mathbb{P}$ and a sequence $\boldsymbol \alpha_\mathbb{P}^m$, 
consider the set of indicator functions $\phi_{i,\mathbb{P}} = \indicator{A_i(\mathbb{P})}:X\rightarrow \{0,1\}$ of the $\alpha$-level sets $A_i(\mathbb{P}) = S_{\alpha_i}(f_\mathbb{P}) - S_{\alpha_{i+1}}(f_\mathbb{P})$ for $i \in \{1,\dots,m-1\}$. Define $ f_m(x)= \sum_{i=1}^{m}\alpha_{i}\phi_{i,\mathbb{P}}(x)$. Then:
$$ \lim_{m\rightarrow \infty} f_{m}(x)= f_\mathbb{P}(x),$$
\noindent
where the convergence is pointwise almost everywhere. Moreover, as the sequence $f_{m}$ is monotonically increasing ($f_{m-1}\leq f_m$), by Dini's Theorem, the convergence is also uniform (converge uniformly almost everywhere). 
\end{theorem}

\begin{corollary}{\textbf{($\alpha$-level sets identification of PMs).}} \label{th2}
If the set of test functions $\mathcal{D}$ contains the indicator functions of the $\alpha$-level sets, then $\mathcal{D}$ is rich enough to discriminate among PMs.
\end{corollary}

Now we elaborate on the construction of a metric that is able to identify PM. Denote by $\mathscr{D}_X$ to the set of probability distributions on $X$ and given a suitable sequence of non-decreasing values $\{\alpha_i \}_{i=1}^m$, define: $\mathscr{D}_X \xrightarrow{\phi_i} \mathcal{D}: \phi_i(\mathbb{P})=\indicator{A_i(\mathbb{P})}$. We propose distances of the form $\sum_{i=1}^{m-1}w_i
d\left(\phi_i(\mathbb{P}),\phi_i(\mathbb{Q})\right)$. Consider, as an example, the measure of the standardized symmetric difference:
$$d\left(\phi_i(\mathbb{P}),\phi_i(\mathbb{Q})\right) =  \frac{\mu\left(A_i(\mathbb{P}) \bigtriangleup A_i(\mathbb{Q})\right)}{\mu\left(A_i(\mathbb{P}) \cup A_i(\mathbb{Q})\right)}.$$

This motivates the definition of the $\alpha$-level set semi-metric as follows.
\begin{definition}{\textbf{(Weighted $\alpha$-level set semi-metric).}}
Given $m \in  \mathbb{N}$, consider two sequences: $\boldsymbol \alpha_\mathbb{P}^m$ and $\boldsymbol \beta_\mathbb{Q}^m$, for $\mathbb{P}$ and $\mathbb{Q}$ respectively. Then define a family of weighted $\alpha$-level set distances between $\mathbb{P}$ and $\mathbb{Q}$ by
\begin{eqnarray}\label{eq:alpha-distance}
d_{{\boldsymbol{\alpha,\beta}}}(\mathbb{P},\mathbb{Q}) &=& \sum_{i=1}^{m-1}w_i
d\left(\phi_i(\mathbb{P}),\phi_i(\mathbb{Q})\right)\\ \nonumber
& =&  \sum_{i=1}^{m-1} w_i \frac{\mu\left(A_i(\mathbb{P}) \bigtriangleup A_i(\mathbb{Q})\right)}{\mu\left(A_i(\mathbb{P}) \cup A_i(\mathbb{Q})\right)},
\end{eqnarray}  
where $w_i\dots,w_{m-1} \in \mathbb{R}^+$ and $\mu$ is the ambient measure.
\end{definition}

Equation \eqref{eq:alpha-distance} can be
interpreted as a weighted sum of Jaccard distances between the $A_i(\mathbb{P})$ and $A_i(\mathbb{Q})$ sets. For $m \gg 0$, when $\mathbb{P} \approx \mathbb{Q}$, then  $d_{{\boldsymbol{\alpha,\beta}}}(\mathbb{P},\mathbb{Q}) \approx 0$ since $\mu\left(A_i(\mathbb{P}) \bigtriangleup A_i(\mathbb{Q})\right) \approx 0$ for all $i \in \{1,\dots,m\}$ (assume $|f_{\mathbb{P}}(x) - f_{\mathbb{Q}}(x)|\leq \varepsilon$ for all $x$, since $f_{\mathbb{P}} \stackrel{\varepsilon \rightarrow 0}{=} f_{\mathbb{Q}}$ then $\mu\left(A_i(\mathbb{P}) \bigtriangleup A_i(\mathbb{Q})\right) \stackrel{\varepsilon \rightarrow 0}{\longrightarrow} 0\; \forall i$, because otherwise contradicts the fact that $f_{\mathbb{P}} \stackrel{\varepsilon \rightarrow 0}{=} f_{\mathbb{Q}}$).

\begin{proposition}\textbf{(Convergence of the $\alpha$-level set semi-metric to a metric).}
$d_{{\boldsymbol{\alpha,\beta}}}(\mathbb{P},\mathbb{Q})$ converges to a metric when $m\rightarrow \infty$.
\end{proposition}

The semi-metric proposed in Eq. \eqref{eq:alpha-distance} obeys the following properties:  is non-negative, that is $d_{{\boldsymbol{\alpha,\beta}}}(\mathbb{P},\mathbb{Q})\geq 0$ and $\lim\limits_{m \rightarrow \infty} d_{{\boldsymbol{\alpha,\beta}}}(\mathbb{P},\mathbb{Q})=0$ if and only if $\mathbb{P}=\mathbb{Q}$. For fixed pairs $(\boldsymbol{\alpha},\mathbb{P})$ and $(\boldsymbol{\beta},\mathbb{Q})$ it is symmetric $d_{{\boldsymbol{\alpha,\beta}}}(\mathbb{P},\mathbb{Q})=d_{{\boldsymbol{\beta,\alpha}}}(\mathbb{Q},\mathbb{P})$. Therefore constitutes a proper metric when $m\rightarrow \infty$. The semi-metric proposed in Eq. \eqref{eq:alpha-distance} is invariant under affine transformations (see the Appendix B for a formal proof). In section \ref{sec:weightch} we will propose a weighting scheme for setting the weights $\{w_i\}_{i=1}^{m-1}$.

Of course, we can calculate $d_{\boldsymbol{\alpha,\beta}}$ in Eq. (\ref{eq:alpha-distance}) only when we know the distribution function for both PMs $\mathbb{P}$ and $\mathbb{Q}$. In practice there will be available two data samples generated from $\mathbb{P}$ and  $\mathbb{Q}$, and we need to define some plug in estimator: Consider estimators $\hat A_i(\mathbb{P}) = \hat S_{\alpha_i}(f_\mathbb{P}) - \hat S_{\alpha_{i+1}}(f_\mathbb{P})$ (details in subsection \ref{sec:oneclass}), then we can estimate
$d_{\boldsymbol{\alpha,\beta}}(\mathbb{P},\mathbb{Q})$ by
\begin{equation} \label{eq:alpha-distance-plugin}
\hat d_{\boldsymbol{\alpha,\beta}}(\mathbb{P},\mathbb{Q}) = \sum_{i=1}^{m-1}w_i \frac{\mu\left(\hat A_i(\mathbb{P}) \bigtriangleup \hat A_i(\mathbb{Q})\right)}{\mu\left(\hat A_i(\mathbb{P}) \cup  \hat A_i(\mathbb{Q})\right)}.
\end{equation}

It is clear that $\mu\left(\hat A_i(\mathbb{P}) \cup  \hat A_i(\mathbb{Q})\right)$ equals the total number of points in $\hat A_i(\mathbb{P}) \cup  \hat A_i(\mathbb{Q})$, say $\#\left(\hat A_i(\mathbb{P}) \cup  \hat A_i(\mathbb{Q})\right)$. 
Regarding the numerator in Eq. (\ref{eq:alpha-distance-plugin}), given two level sets, say $A$ and $B$ to facilitate the notation, and the corresponding sample estimates $\hat A$ and $\hat B$, one is tempted to estimate $\mu(A \bigtriangleup B)$, the area of region $A \bigtriangleup B$, by $\widehat{\mu(A\bigtriangleup B)} = \#(\hat A-\hat B) \cup \#(\hat B-\hat A) = \#(A\cup B) - \#(A \cap B)$. However this is incorrect since probably there will be no points in common between  $\hat A$ and $\hat B$ (which implies $\widehat{A\bigtriangleup B} = \widehat{A \cup B}$).

In our particular case, the algorithm in Table 1 shows that $\hat A_i(\mathbb{P})$ is always a subset of the sample $s_\mathbb{P}$ drawn from the density function $f_{\mathbb{P}}$, and we will denote this estimation by $s_{\hat A_i(\mathbb{P})}$ from now on. We will reserve the
notation $\hat{A_i}(\mathbb{P})$ for the covering estimation of $A_i(\mathbb{P})$ defined by
$\cup_j^n B(x_j,r_A)$ where $x_j \in s_{\hat A_i(\mathbb{P})}$, $B(x_j,r_A)$ are closed balls with centres at 
$x_j$ and (fixed) radius $r_A$ \citep{bib:devroye}. The radius is chosen to be constant (for data points in $\hat{A_i}(\mathbb{P})$)  because we can assume that density is approximately constant inside region $\hat{A_i}(\mathbb{P})$, if the partition $\{\alpha_{i}\}_{i=1}^m$ of the set is fine enough. For example, in the experimental section, we fix $r_A$ as the median distance between the points that belongs to the set $s_{\hat A_i(\mathbb{P})}$.

\begin{figure*}[t!]
  \begin{center}
    \includegraphics[height=5cm,width=14cm]{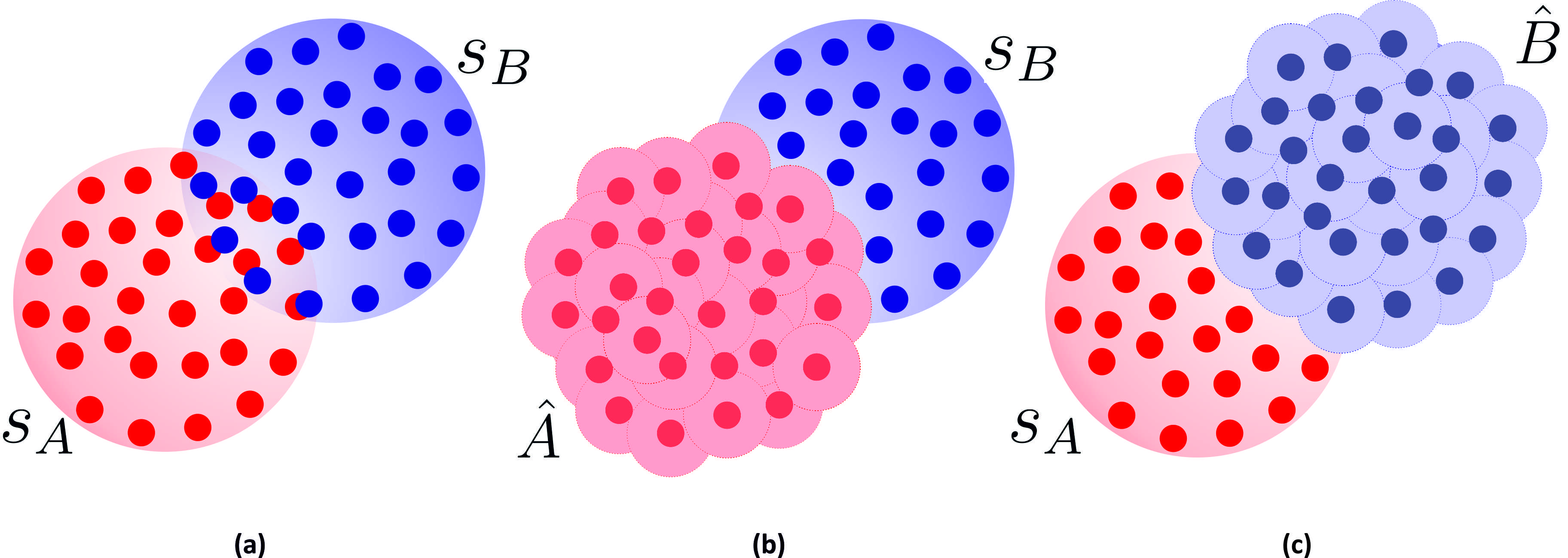}
  \end{center}
      \caption{Set estimate of the symmetric difference. (a) Data samples $s_A$ (red) and $s_B$ (blue). (b) $s_B$ - Covering $\hat A$: blue points. (c) $s_A$ - Covering $\hat B$: red points. Blue points in (b) plus red points in (c) are the estimate of $A\bigtriangleup  B$.}\label{fig:difsimetrica}
\end{figure*}

To illustrate this notation we include Figure \ref{fig:difsimetrica}. In Figure \ref{fig:difsimetrica} (a) we show two data different samples from $\alpha$-level sets $A$ and $B$: $s_{A}$ (red points) and $s_{B}$ (blue points), respectively. In Figure \ref{fig:difsimetrica}(b) $\hat A$ is the covering estimation of set $A$ made up of the union of balls centered in the data red points $s_{A}$. 
That is $\hat{A} = \cup_j^n B(x_j,r_A) \xrightarrow[n\to\infty]{r_A\to0} A$.  Figure \ref{fig:difsimetrica} (c) can be interpreted equivalently regarding the covering of the sample $s_{B}$. The problem of calculating $\widehat{\mu(A\bigtriangleup B)}$ thus reduces to estimate the number points in $\hat B$ not belonging to the covering estimate of $A$, plus the number points in $\hat A$ not belonging to the covering estimate of $B$. To make the computation explicit consider $x\in A, y\in B$ and define 
\begin{eqnarray}\nonumber
I_{r_A,r_B}(x,y) &=& \indicator{B(x,r_A)}(y) + \indicator{B(y,r_B)}(x)\\ \nonumber
&-& \indicator{B(x,r_A)}(y) \indicator{B(y,r_B)}(x),
\end{eqnarray}
where $I_{r_A,r_B}(x,y)=1$ when $y$ belongs to the covering $\hat{A}$, 
$x$ belongs to the covering $\hat{B}$ or both events happen. Thus if we 
define
$$I(A,B) = \sum_{x\in A}\sum_{y\in B} I_{r_A,r_B}(x,y),$$
we are able to estimate the symmetric difference by
$$ 
\widehat{\mu(A\bigtriangleup B)} = \widehat{\mu(A\cup B)} -\widehat{\mu(A\cap B)}=  \#{\mu(A\cup B)} - I(A,B).$$

\subsection{Estimation of level sets}
\label{sec:oneclass}

To estimate the set denoted as  $s_{\hat A_i(\mathbb{P})}$ we implement a One-Class Neighbor Machine approach \citep{bib:munoz2006,bib:naiveoneclass}. The  One-Class Neighbor Machine solves the following optimization problem:
\begin{equation}\label{eq:snmprimal1}
\begin{array}{lll}
\displaystyle \max_{\rho,\xi} & \displaystyle   \nu n \rho - \sum_{i=1}^n \xi_i &\\
\mbox{s.t.} & g(x_i)\ge \rho-\xi_i \, ,& \\
& \xi_i \ge 0 ,& i=1,\ldots,n \, ,
\end{array}
\end{equation}

where $g(x) = M(x,s_n)$ is a sparsity measure (see the Appendix C for further details), $\nu \in [0,1]$ such that $P(S_\alpha)=1-\nu$, $\xi_i$ with $i=1,\dots,n$ are slack variables and $\rho$ is a predefined constant. With the aid of the Support Neighbor Machine, we estimate a density contour cluster $S_{\alpha_i}(f)$ around the mode for a suitable sequence of values $\{\nu_i\}_{i=1}^m$ (note that the sequence $0\geq \nu_1, \dots , \nu_m = 1$  it is in a one-to-one correspondence with the sequence $0\leq \alpha_1 < \ldots < \alpha_m=\max_{x\in X} f_\mathbb{P}(x)$). In Table 1 we present the algorithm to estimate $S_\alpha(f)$ of a density function $f$. Hence, we take $ s_{\hat A_i(\mathbb{P})} = \hat S_{\alpha_i}(f_\mathbb{P}) - \hat S_{\alpha_{i+1}}(f_\mathbb{P})$, where $\hat S_{\alpha_i}(f_\mathbb{P})$ is estimated by $R_n$ defined in Table 1 (the same estimation procedure applies for $s_{\hat A_i(\mathbb{Q})}$).

\begin{table*}
\begin{tabular}{ l p{16cm} }
 & \textbf{Estimation of} $\mathbf{R_n = \hat S_{\alpha}(f)}$:\\
\hline
\textbf{1} & Choose a constant $\nu \in [0,1]$. \\
\textbf{2} & Consider the order induced in the sample $s_n$ by the sparsity measure $g_n(x)$,
             that is, $g_n(x_{(1)}) \le \cdots \le g_n(x_{(n)})$, where $x_{(i)}$ denotes the $i^{th}$ sample, ordered                 
             after $g$.\\
\textbf{3} & Consider the value $\rho_n^* = g(x_{(\nu n)})$ if $\nu n \in \mathbb{N}$,
             $\rho_n^* = g_n(x_{([\nu n]+1)})$ otherwise,
             where $[x]$ stands for the largest integer not greater than $x$.\\
\textbf{4} & Define $h_n(x) = sign(\rho_n^* - g_n(x))$.\\
\hline
\end{tabular}
\label{tab:algorit}
\caption{Algorithm to estimate minimum volume sets ($S_{\alpha}(f)$) of a density $f$.}
\end{table*}

The computational complexity of the algorithm of Table 1 and more details on the estimation of the regions $\mathbf{\hat S_{\alpha}(f)}$ are contained in \citep{munoz2004one,bib:naiveoneclass,bib:munoz2006}. The execution time required to compute $\mathbf{\hat S_{\alpha}(f)}$ grows at a rate of order $\mathcal{O}(dn^2)$, where $d$ represent the dimension and $n$ the sample size of the data at hand. 

\subsection{Choice of weights for $\alpha$-level set distances}\label{sec:weightch}

In this section we define a weighting scheme for the family of distances defined by Eq. \eqref{eq:alpha-distance}. Denote by $s_{\mathbb{P}}$ and 
$s_{\mathbb{Q}}$ the data samples corresponding to PMs $\mathbb{P}$ and $\mathbb{Q}$ respectively, and denote by $s_{\hat A_i(\mathbb{P})}$ and $s_{\hat A_i(\mathbb{Q})}$ the data samples that estimate $A_i(\mathbb{P})$ and $A_i(\mathbb{Q})$,
respectively. Remember that we can estimate these sets by coverings
$\hat A_i(\mathbb{P}) = \cup_{x\in s_{\hat A_i(\mathbb{P})}} B(x,r_{\hat A_i(\mathbb{P})})$,
$\hat A_i(\mathbb{Q}) = \cup_{x\in s_{\hat A_i(\mathbb{Q})}} B(x,r_{\hat A_i(\mathbb{Q})})$.

 Let $m$ denote the size of the $\boldsymbol \alpha_\mathbb{P}^m$ and $\boldsymbol \beta_\mathbb{Q}^m$ sequences. Denote by
$n_{\hat A_i(\mathbb{P})}$ the number of data points in $s_{\hat A_i(\mathbb{P})}$, $n_{\hat A_i(\mathbb{Q})}$ the number of data points in $s_{\hat A_i(\mathbb{Q})}$,
$r_{\hat A_i(\mathbb{P})}$ the (fixed) radius for the covering  $\hat A_i(\mathbb{P})$ and $r_{\hat A_i(\mathbb{Q})}$ the (fixed) radius for the covering  $\hat A_i(\mathbb{Q})$, usually the mean or the median distance inside the region $\hat A_i(\mathbb{P})$ and $\hat A_i(\mathbb{Q})$ respectively. We define the following weighting scheme:

\begin{eqnarray}
\nonumber
w_i & = &\frac{1}{m}\sum\limits_{x\in s_{\hat A_i(\mathbb{P})}}^{n_{\hat A_i(\mathbb{P})}}\sum\limits_{y\in
s_{\hat A_i(\mathbb{Q})}}^{n_{\hat A_i(\mathbb{Q})}} \left( 1- I_{r_{\hat A_i(\mathbb{P})},r_{\hat A_i(\mathbb{Q})}}(x,y)\right) \cdot \\
&  & \frac{\parallel x - y\parallel_2}{ (s_{\hat A_i(\mathbb{Q})} - \hat A_i(\mathbb{P})) \cup (s_{\hat A_i(\mathbb{P})}- \hat A_i(\mathbb{Q}))}.
\end{eqnarray}

The weight $w_i$ is a weighted average of distances between a point of $s_{\hat A_i(\mathbb{P})}$ and a point of $s_{\hat A_i(\mathbb{Q})}$ where $\|x-y\|_2$ is taken into account only when $I_{r_{\hat A_i(\mathbb{P})},r_{\hat A_i(\mathbb{Q})}}(x,y)=0$. More details about the weighting scheme and its extension can be seen in the Appendix D.

\section{Experimental work}\label{sec:exp}
Since the proposed distance is intrinsically nonparametric, there are no simple parameters on which we can concentrate our attention to do exhaustive benchmarking. The strategy will be to compare the proposed distance to other classical PM distances for some well known (and parametrized) distributions and for real data problems. Here we consider distances belonging to the main types of PMs metrics: Kullback-Leibler (KL) divergence \citep{bib:boltz,bib:nguy1} ($f$-divergence and also Bregman divergence), t-test (T) measure (Hotelling test in the multivariate case), Maximum Mean Discrepancy (MMD) distance \citep{gretton2012kernel,Sriperumb_HSEmbsd} and Energy distance \citep{bib:rizzo,bib:Sejdin} (an Integral Probability Metric, as it is demonstrated in \citep{bib:Sejdin}).

\subsection{Artificial data}
\subsubsection{Discrimination between normal distributions}
In this experiment we quantify the ability of the considered PM distances to test the null hypothesis $H_0:\mathbb{P} = \mathbb{Q}$ when $\mathbb{P}$ and $\mathbb{Q}$ are multivariate normal distributions. To this end, we generate a data sample of size $100  d$ from a normal distribution $N(\mathbf{0},\mathbf{I}_d)=\mathbb{P}$, where $d$ stands for dimension and then we generate 1000 \textit{iid} data samples of size $100  d$ from the same $N(\mathbf{0},\mathbf{I}_d)$ distribution. Next we calculate the distances between each of these $1000$ \textit{iid} data samples and the first data sample to obtain the $95\%$ distance percentile denoted as $d_{H_0}^{\;95\%}$. 

Now define $\boldsymbol{\delta} =\delta \mathbf{1}= \delta (1,\dots,1) \in \mathbb{R}^d$ and increase $\delta$ by small amounts (starting from 0). For each $\boldsymbol{\delta}$ we generate a data sample of size $100 d$ from a $N(\mathbf{0}+\boldsymbol{\delta},\mathbf{I}_d)=\mathbb{Q}$ distribution. If $d(\mathbb{P},\mathbb{Q})>d_{\mathbb{P}}^{\;95\%}$ we conclude that the present distance is able to discriminate between both populations (we reject $H_0$) and this is the value $\delta^*$ referenced in Table \ref{tab:table1}. To track the power of the test, we repeat this process 1000 times and fix $\delta^*$ to the present $\delta$ value if the distance is above the percentile in $90\%$ of the cases.
Thus we are calculating the minimal value $\mathbf{\delta^*}$ required for each metric in order to discriminate between populations with a $95\%$ confidence level (type I error $= 5\%$) and a $90\%$ sensitivity level (type II error $= 10\%$). In Table \ref{tab:table1} we report the minimum distance ($\delta^* \sqrt{d})$ between distributions centers required to discriminate for each metric in several alternative dimensions, where small values implies better results. In the particular case of the $T$-distance for normal distributions we can use the Hotelling test to compute a $p$-value to fix the $\delta^*$ value.

\begin{table*}[t!]
\begin{center}
\caption{$\mathbf{\delta^*}\sqrt{d}$ for a 5\% type I and 10\% type II errors.} 
\begin{tabular}{|l|ccccccccccc|}  
\hline                            
Metric&d:	&1 &2 &3 &4 &5 &10 &15 &20 &50 &100 \\[0.5ex]
\hline
KL &    & $0.870$ & $0.636$ & $0.433$ & $0.430$ &  
          $0.402$ & $0.474$ & $0.542$ & $0.536$ & $0.495$ & $0.470$\\ 
T  &    & $0.490$ & $0.297$ & $0.286$ & $0.256$& 
          $0.246$ & $0.231$ & $0.201$ & $0.182$& 
          $0.153$ & $0.110$ \\
Energy && $0.460$ & $0.287$ & $0.284$ & $0.256$&
          $0.250$ & $0.234$ & $0.203$ & $0.183$& 
          $0.158$ & $0.121$\\
MMD    && $0.980$ & $0.850$ & $0.650$ & $0.630$&
          $0.590$ & $0.500$ & $0.250$ & $0.210$& 
          $0.170$ & $0.130$\\
LS(0) & & $0.490$ & $0.298$ & $0.289$ & $0.252$& 
          $0.241$ & $0.237$ & $0.220$ & $0.215$& 
          $0.179$ & $0.131$ \\
LS(1)  && $\mathbf{0.455}$ & $\mathbf{0.283}$ & $\mathbf{0.268}$ & $\mathbf{0.240}$&
          $\mathbf{0.224}$ & $\mathbf{0.221}$ & $\mathbf{0.174}$ & $\mathbf{0.178}$& 
          $\mathbf{0.134}$ & $\mathbf{0.106}$ \\
\hline
\end{tabular}
\end{center}
\label{tab:table1}
\end{table*}
The data chosen for this experiment are ideal for the use of the $T$ statistics that, in fact, outperforms KL and MMD. However, Energy distance works even better than $T$ distance in dimensions 1 to 4. The LS(0) distance work similarly to $T$ and Energy until dimension $10$. The LS(1) distance outperform to all the competitor metrics in all the considered dimensions.

In a second experiment we consider again normal populations but different variance-covariance matrices. Define as an expansion factor $\sigma \in \mathbb{R}$ and increase $\sigma$ by small amounts (starting from 0) in order to determine the smallest $\sigma^*$ required for each metric in order to discriminate between the $100d$ sampled data points generated for the two distributions: $N(\mathbf{0},\mathbf{I}_d)=\mathbb{P}$ and $N(\mathbf{0},(1+\sigma)\mathbf{I}_d)=\mathbb{Q}$. If $d(\mathbb{P},\mathbb{Q})>d_{\mathbb{P}}^{\;95\%}$ we conclude that the present distance is able to discriminate between both populations and this is the value $(1+\sigma^*)$ reported in Table \ref{tab:table2}. To make the process as independent as possible from randomness we repeat this process 1000 times and fix $\sigma^*$ to the present $\sigma$ value if the distance is above the $90\%$ percentile of the cases, as it was done in the previous experiment.

\begin{table*}[t1]
\begin{center}
\caption{$(1+\sigma^*)$ for a 5\% type I and 10\% type II errors.}  
\begin{tabular}{|l|ccccccccccc|}  
\hline                
Metric&dim:	&1 &2 &3 &4 &5 &10 &15 &20 &50 &100 \\[0.5ex]
\hline
KL &&        $3.000$ & $1.700$ & $1.250$ & $1.180$ & 
             $1.175$ & $1.075$ & $1.055$& $1.045 $ &
             $1.030$ & $1.014$\\ 
T  &&        $-$ & $-$ & $-$ & $-$&              
             $-$ & $-$ & $-$ & $-$&
             $-$ & $-$ \\
Energy &&    $1.900$ & $1.600$ & $1.450$ & $1.320$&
             $1.300$ & $1.160$ & $1.150$ & $1.110$&
             $1.090$ & $1.030$\\
MMD    && $6.000$ & $4.500$ & $3.500$ & $2.900$&
          $2.400$ & $1.800$ & $1.500$ & $1.320$& 
          $1.270$ & $1.150$ \\
LS(0) &&     $1.850$ & $1.450$ & $1.300$ & $1.220$& 
             $1.180$ & $1.118$ & $1.065$ & $1.040$&
             $1.030$ & $1.012$ \\
LS(1) &&     $\mathbf{1.700}$ & $\mathbf{1.350}$ & $\mathbf{1.150}$ & $\mathbf{1.120}$&
             $\mathbf{1.080}$ & $\mathbf{1.050}$ & $\mathbf{1.033}$ & $\mathbf{1.025}$& 
             $\mathbf{1.015}$ & $\mathbf{1.009}$ \\
\hline
\end{tabular}
\end{center}
\label{tab:table2}
\end{table*}
There are no entries in Table \ref{tab:table2} for the T distance because it was not able to distinguish between the considered populations in none of the considered dimensions. The MMD distance do not show a good discrimination power in this experiment. We can see here again that the proposed LS(1) distance is better than the competitors in all the dimensions considered, having the LS(0) and the KL similar performance in the second place among the metrics with best discrimination power.

\subsubsection{Homogeneity tests}
This experiment concerns a homogeneity test between two populations: a mixture between a Normal and a Uniform distribution ($\mathbb{P}=\alpha N(\mu=1,\sigma=1) + (1-\alpha)U(a=1,b=8)$ where $\alpha = 0.7$) and a Gamma distribution ($\mathbb{Q}=\gamma(shape = 1, scale=2)$). To test the null hypothesis: $H_0: \mathbb{P}=\mathbb{Q}$ we generate two random i.i.d. samples of size $100$ from $\mathbb{P}$ and $\mathbb{Q}$, respectively. Figure \ref{fig:mixtasimet} shows the corresponding density functions for $\mathbb{P}$ and $\mathbb{Q}$.
\begin{figure}[t!]
  \begin{center}
    \includegraphics[scale=.45]{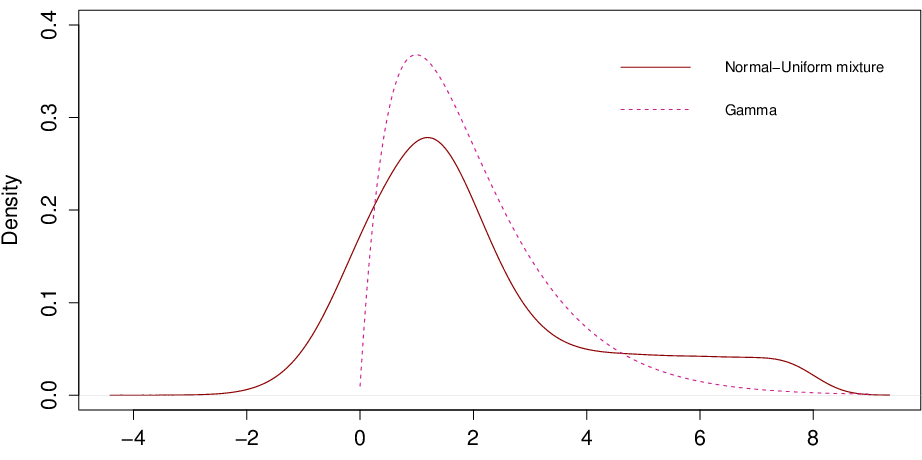}
  \end{center}
      \caption{Mixture of a Normal and a Uniform Distribution and a Gamma distribution.} \label{fig:mixtasimet}
\end{figure}
In the cases of KL-divergence, T, Energy MMD and LS distances we proceed as in the previous experiment, we run a permutation test based on $1000$ random permutation of the original data in order to compute the $p$-value. In the case of Kolmogorov-Smirnov, $\chi^2$ and Wilcoxon  test we report the $p$-value given by these tests. Results are displayed in Table \ref{tab:table3}:  Only the LS distances are able to distinguish between both distributions. Notice that first and second order moments for both distribution are quite similar in this case ($\mu_{\mathbb{P}}=2.05 \simeq \mu_{\mathbb{Q}}=2$ and $\sigma_{\mathbb{P}}=4.5 \simeq \sigma_{\mathbb{Q}}=4$) and additionally both distributions are strongly asymmetric, which also contributes to explain the failure of those metrics strongly based on the use of the first order moments.
\begin{table}[h]
\caption{Hypothesis test between a mixture of Normal and Uniform distributions and a Gamma distribution.}
\centering  
\begin{tabular}{|l|ccc|}  
\hline
Metric&Parameters &$p$-value &Reject? \\[0.5ex]
\hline
Kolmogorov-Smirnov    & &   $0.281$     &No.\\ 
$\chi^2$ test         & &   $0.993$ &No.\\
Wilcoxon test         & &  $0.992$      &No.\\
KL 		&$k = 10$& $0.248$                        &No.\\
T  		&  &  $0.342$               &No.\\
Energy  &  & $0.259$                             &No. \\
MMD &   & $ 0.177$                             &  No.\\
LS (0)  &$m=10$& $\mathbf{0.092}$                  &\textbf{Yes.}\\
LS (1)  &$m=10$& $\mathbf{0.050}$                  &\textbf{Yes.}\\
\hline
\end{tabular}
\label{tab:table3}
\end{table}

\subsection{Two real case-studies}

\subsubsection{Shape classification}
As an application of the preceding theory to the field of pattern recognition problem we consider the MPEG7 CE-Shape-1 \citep{bib:mpeg7}, a well known shape database. We select four different classes of objects/shapes from the database: hearts, coups, hammers and bones. For each object class we choose $3$ images in the following way: $2$ standard images plus an extra image that exhibit some distortion or rotation (12 images in total). In order to represent each shape we do not follow the usual approach in pattern recognition that consists in representing each image by a feature vector catching its relevant shape aspects; instead we will look at the image as a cloud of points in $\mathbb{R}^2$, according to the following procedure: Each image is transformed to a binary image where each pixel assumes the value $1$ (white points region) or $0$ (black points region) as in Figure \ref{fig:procesoobjeto} (a). For each image $i$ of size $N_i \times M_i$ we generate a uniform sample of size $N_i M_i$ allocated in each position of the shape image $i$. To obtain the cloud of points as in Figure \ref{fig:procesoobjeto} (b) we retain only those points which fall into the white region (image body) whose intensity gray level are larger than a variable threshold fixed at $0.99$ so as to yield around one thousand and two thousand points image representation depending on the image as can be seen in Figure \ref{fig:procesoobjeto} (b). 
\begin{figure}[t!]
  \begin{center}
    \includegraphics[height=5cm,width=8cm]{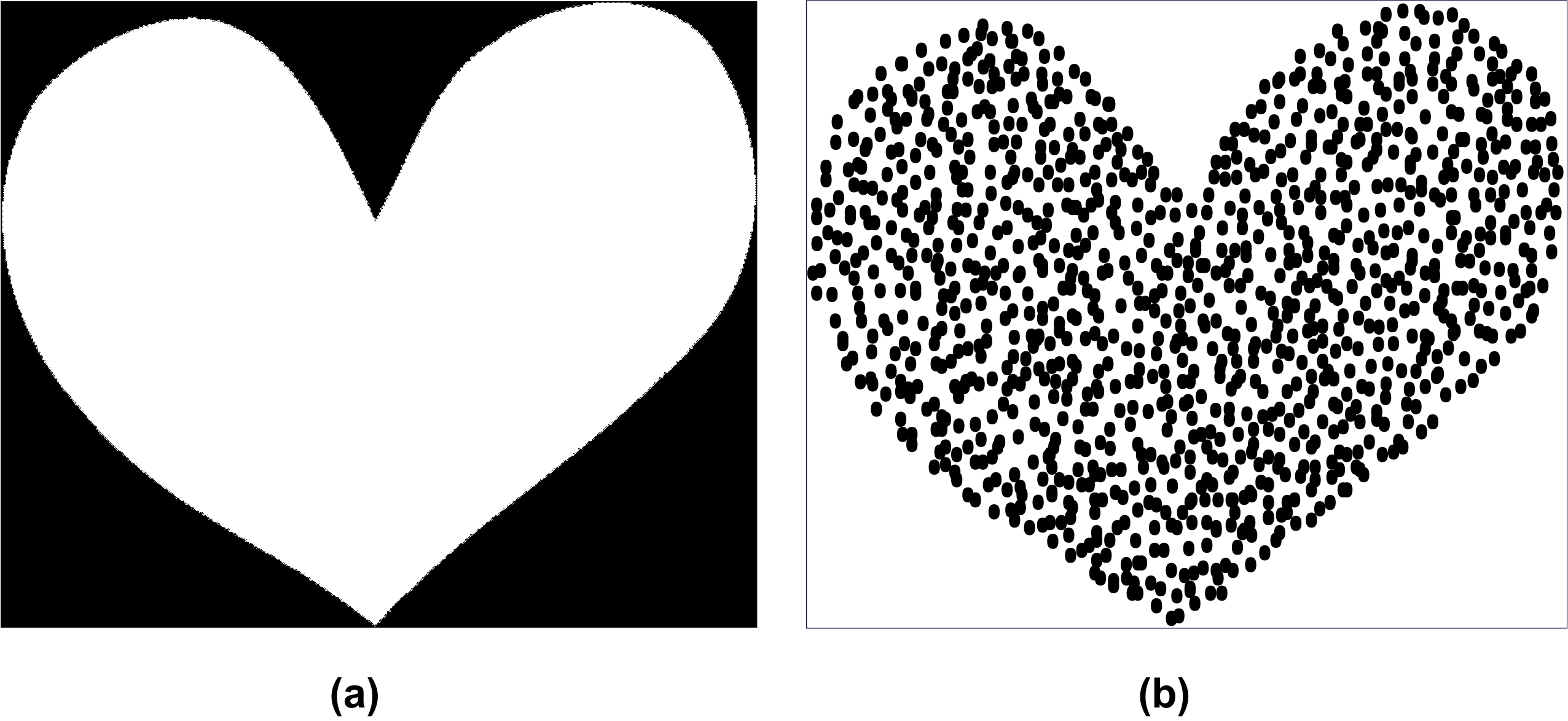}
  \end{center}
      \caption{Real image (a) and sampled image (b) of a hart in the MPEG7 CE-Shape-1 database.} \label{fig:procesoobjeto}
\end{figure}
\noindent
After rescaling and centering, we compute the $12 \times 12$ image distance matrices, using the LS(1) distance and the KL divergence, and then compute Euclidean coordinates for the images via MDS  (results in Figure \ref{fig:objectsclustering}). It is apparent that the LS distance
produces a MDS map coherent with human image perception (fig. \ref{fig:procesoobjeto} (a)). This does not happen for the rest of tested metrics, in particular for the KL divergence as it is shown in Figure \ref{fig:procesoobjeto} (b)). 
\begin{figure*}[t!]
  \begin{center}
    \includegraphics[height=6cm,width=14cm]{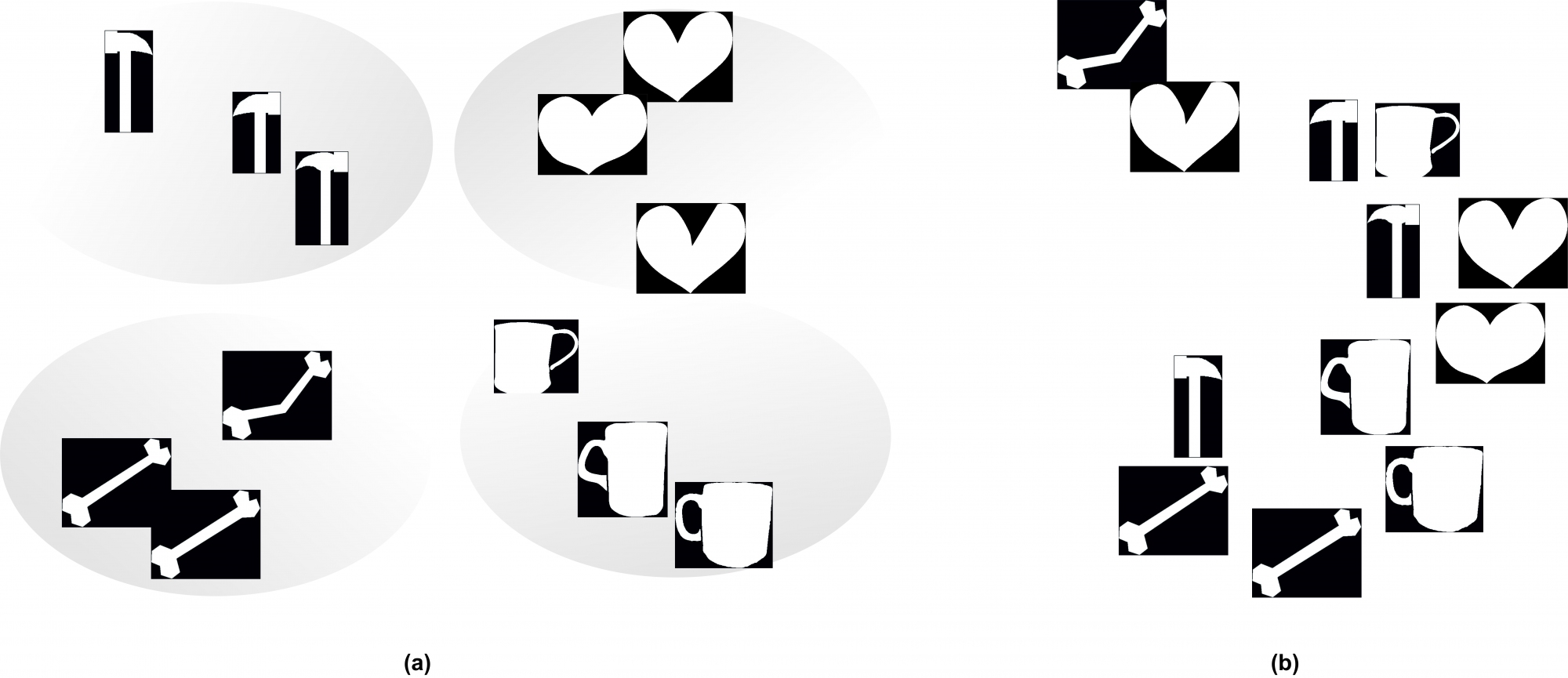}
  \end{center}
      \caption{Multi Dimensional Scaling representation for objects based on (a) LS(1) and (b) KL divergence.} \label{fig:objectsclustering}
\end{figure*}
\subsubsection{Testing statistical significance in Microarray experients} \label{sec:micro}

\begin{figure}[t!]
\begin{center}
\includegraphics[scale=0.9]{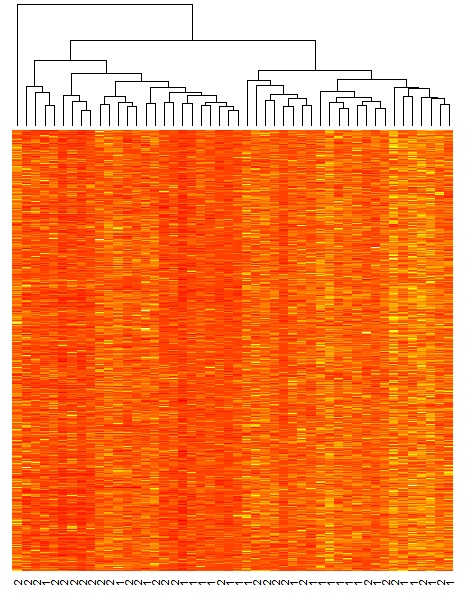}
\caption{Affymetrix U133+2 micro-arrays data from the post trauma recovery experiment. On top, a hierarchical cluster of the patients using the Euclidean distance is included. At the bottom of the plot the grouping of the patients is shown: 1 for  ``early recovery" patients  and 2 for ``late recovery" patients.} \label{fig:microarray}
\end{center}
\end{figure}
Here we present an application of the proposed LS distance in the field of Bioinformatics. The data set we analyze comes from an experiment in which the time to respiratory recovery in ventilated post trauma patients is studied.  Affymetrix U133+2 micro-arrays were prepared at days 0, 1, 4, 7, 14, 21 and 28. In this analysis, we focus on a subset of 46 patients  which were originally divided into two groups: ``early recovery patients'' (group $G_1$) that recovered  ventilation prior to day seven and ``late recovery patients " (group $G_2$), those who recovered ventilation after day seven. The size of the groups is 22 and 26 respectively.

It is of clinical interest to find differences between the two groups of patients. In particular, the originally goal of this study was to test the association of inflammation on day one and
subsequent respiratory recovery. In this experiment we will show how the proposed distance can be used in this context to test statistical differences between the groups and also to identify the genes with the largest effect in the post trauma recovery.

From the original data set \footnote{http://www.ncbi.nlm.nih.gov/geo/query/acc.cgi?acc=GSE13488}
we select the sample of 675  probe sets corresponding to those genes whose GO annotation include the term``inflammatory". To do so we use a query (July 2012) on the Affymetrix web site (affymetrix.com). The idea of this search is to obtain a pre-selection of the genes involved in post trauma recovery in order to avoid working with the whole human genome.

Figure \ref{fig:microarray} shows the heat map of day one gene expression for the 46 patients (columns) over the 675 probe-sets. By using a hierarchical procedure, it is apparent that the two main clusters we find do not correspond to the two groups of patients of the experiment. However, the first cluster (on the left hand of the plot) contains mainly patients form the ``early recovery" group (approx. 65 \%) whereas the second cluster (on the right) is mainly made up of patients from the ``late recovery'' group (approx 60\%). This lack of balance suggests a different pattern of gene expression between the two groups of patients.

\begin{figure*}[t!]
\begin{center}
\includegraphics[scale=0.4]{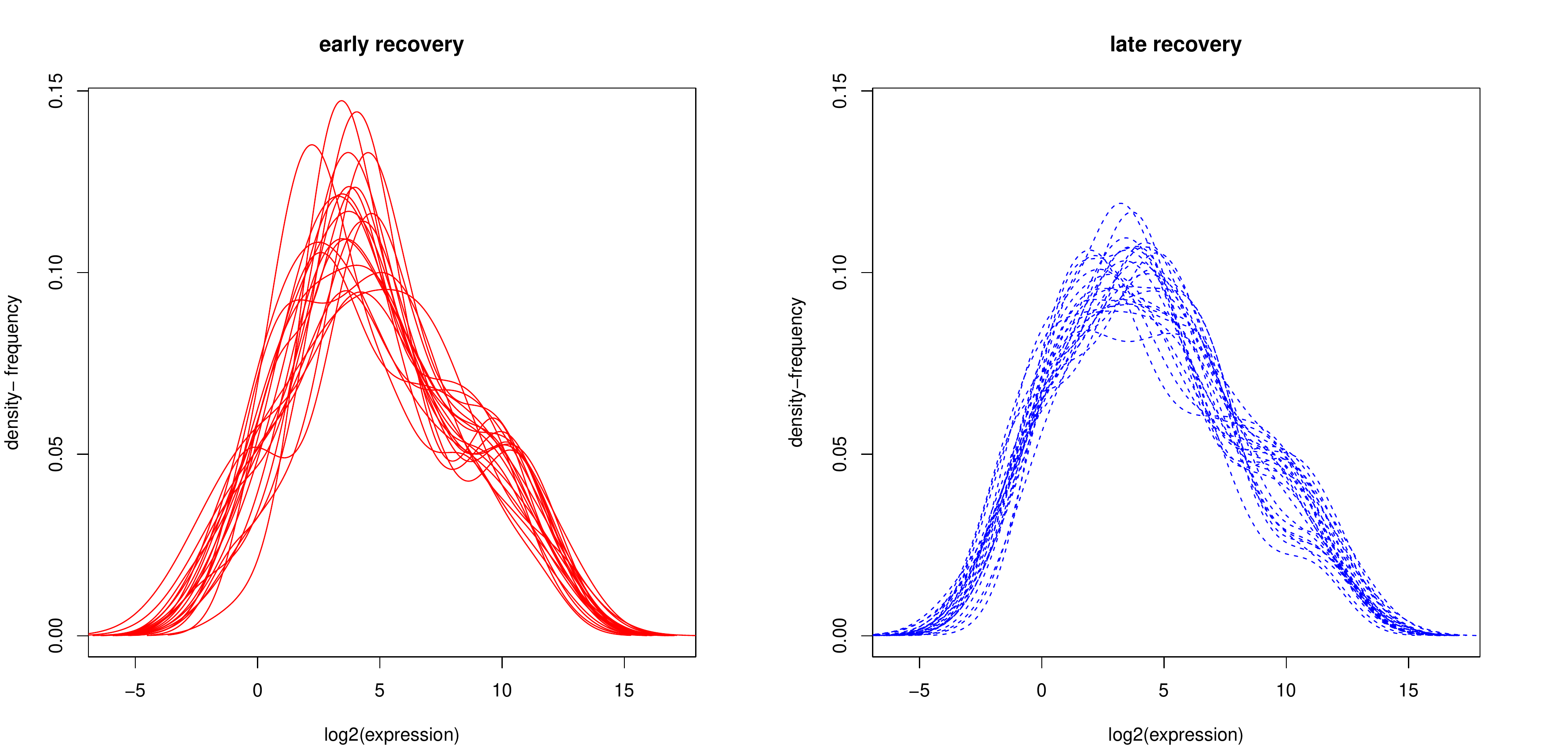}
\caption{Gene density profiles (in logarithmic scale) of the two groups of patients in the sample. The 50 most significant genes were used to calculate the profiles with a kernel density estimator.}\label{fig:profiles}
\end{center}
\end{figure*}

In order to test if statistical differences exists between the groups $G_1$ and $G_2$ we define, inspired by \citep{bib:Hyeden}, an statistical test based on the LS distance proposed in this work. To this end, we identify each patient $i$ with a probability distribution $\mathbb{P}_i$. The expression of the 675 genes across the probe-sets are assumed to be samples of such distributions. Ideally, if the genes expression does not have any effect on the recovery speed then all distributions $\mathbb{P}_i$ should be equal ($H_0$). On the other hand, assume that expression of a gene or a group of genes effectively change between ``early" and ``late" recovery patients. Then, the distributions $\mathbb{P}_i$ will be different between patients belonging to groups $G_1$ and $G_2$ ($H_1$).

To validate or reject the previous hypothesis, consider the proposed LS distance $\hat{d}_{\alpha}(\mathbb{P}_i,\mathbb{P}_j)$  defined in (\ref{eq:alpha-distance-plugin}) for two patients $i$ and $j$. Denote by 
\begin{equation}\nonumber
\Delta_1= \frac{1}{22(22-1)} \sum_{i,j \in G_1} \hat{d}_{\alpha}(\mathbb{P}_i,\mathbb{P}_j),
\end{equation}
\begin{equation}\nonumber
\Delta_2= \frac{1}{26(26-1)} \sum_{i,j\in G_2} \hat{d}_{\alpha}(\mathbb{P}_i,\mathbb{P}_j)
\end{equation}
and
\begin{equation}\nonumber
\Delta_{12}= \frac{1}{22 \cdot 26} \sum_{i\in G1,j\in G_2} \hat{d}_{\alpha}(\mathbb{P}_i,\mathbb{P}_j),
\end{equation}   
the averaged  $\alpha$-level set distances within and between the groups of patients. Using the previous quantities we define a distance between the groups $G_1$ and $G_2$ as
\begin{equation}
\Delta^* = \Delta_{12} - \frac{\Delta_1+ \Delta_2}{2}.
\end{equation}

Notice that if the distributions are equal between the groups then $\Delta^*$ will be close to zero. On the other hand, if the distributions are similar within the groups and different between them, then $\Delta^*$ will be large. To test if $\Delta^*$ is large enough to consider it statistically significant we need the distribution of $\Delta^*$ under the null hypothesis. Unfortunately, this distribution is unknown and some re-sampling technique must be used. In this work we approximate it by calculating a sequence of distances $\Delta^{*} _{(1)},\dots,\Delta^*_{(N)}$ where each $\Delta^*_{(k)}$ is the distance between the groups $G_1$ and $G_2$ under a random permutation of the patients. For a total of $N$ permutations, then 
\begin{equation}
p-value = \frac{\#[\Delta^*_{(k)}\geq \Delta^* :k=1,\dots,N]}{N},
\end{equation}
where $\#[\Theta]$ refers to the number of times the condition $\Theta$ is satisfied, is a one-side p-value of the test. 

\begin{figure*}[t!]
 \begin{center}
  \mbox{
        \subfigure[Heat-map of the top-50 ranked genes (rows). A hierarchical cluster of the patients is included on top. The labels of the patients regarding their recovery group are detailed at the bottom of the plot.]{\includegraphics[height=7.5cm,width=7.5cm]{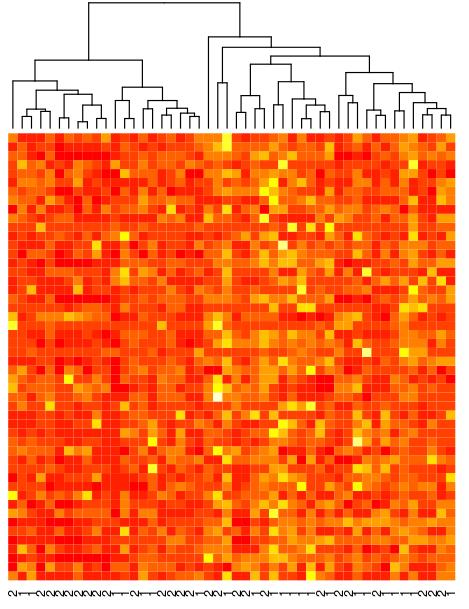}}\quad
   \subfigure[P-values obtained by the proposed  $\alpha$-level set distance based test for different samples of increasing number of genes.]{\includegraphics[height=7.5cm,width=7.5cm]{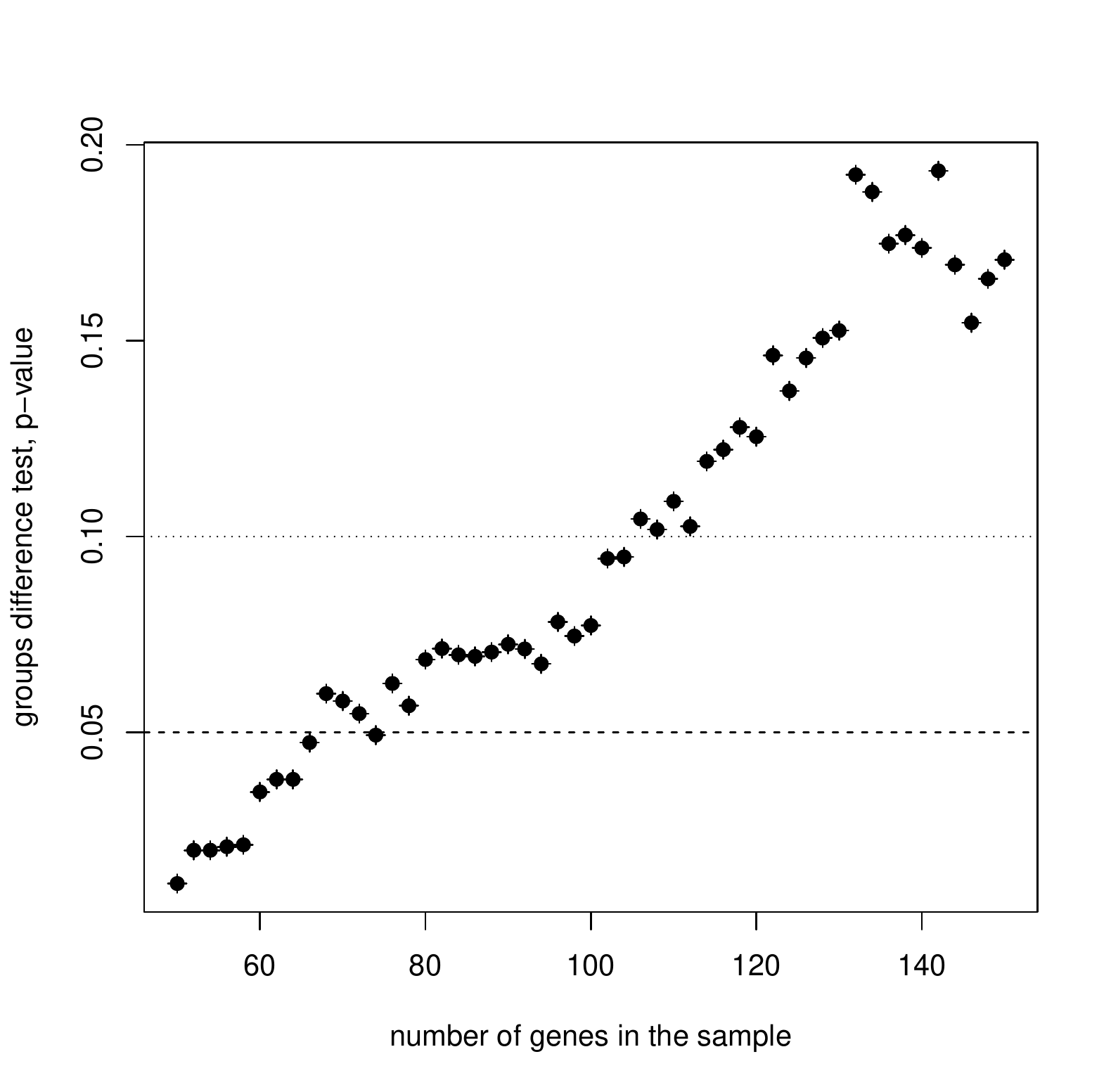}}\quad
       }
\end{center}
\caption{Heat-map of the 50-top ranked genes and p-values for different samples.}\label{figure:heat.and.pvalues}
\end{figure*}

We apply the previous LS distance based test (weighting scheme 1 with 10000 permutations) using the values of the 675 probe-sets and we obtain a p-value $= 0.1893$. This result suggests that none differences exists between the groups exist. The test for micro-arrays proposed in \citep{bib:Hyeden} also confirms this result with a p-value of $0.2016$. The reason to explain this -a priori- unexpected result is that, if differences between the groups exist, they are probably hidden by a main group of genes with similar behavior between the groups. To validate this hypothesis, we first rank the set of 675 genes in terms of their individual variation between groups $G_1$ and $G_2$. To do so, we use the p-values of individual difference mean T-tests. Then, we consider the top-50 ranked genes and we apply the  $\alpha$-level set distance test. The obtained p-value is $0.010$, indicating a significant difference in gene expression of the top-50 ranked genes.  In Figure \ref{fig:profiles} we show the estimated density profiles of the patients using a kernel estimator. It is apparent that the profiles between groups are different as it is reflected in the obtained results. In Figure \ref{figure:heat.and.pvalues}, we show the heat-map calculated using the selection of 50 genes. Note  that a hierarchical cluster using the Euclidean distance, which is the most used technique to study the existence of groups in micro-array data, is not able to accurately reflect the existence of the two groups even for the most influential genes.

To conclude the analysis we go further from the initial 50-genes analysis. We aim to obtain the whole set of genes in the original sample for which differences between groups remain significant. To do so, we sequentially include in the first 50-genes sample the next-highest ranked genes and we apply to the augmented data sets the LS distance based test. The p-values of such analysis are shown in Figure \ref{figure:heat.and.pvalues}. Their value increases as soon as more genes are included in the sample. With a type-I error of 5\%, statistical differences are found for the 75 first genes. For a 10\% type-I error, with the first 110 genes we still are able to find differences between groups. This result shows that differences between ``early'' and ``late'' recovery trauma patients exist and they are caused by the top-110 ranked genes of the Affymetrix U133+2 micro-arrays (filtered by the query ``inflammatory"). By considering each patient as a probability distribution the LS distance has been used to test differences between groups and to identify the most influential genes of the sample. This shows the ability of the new proposed distance to provide new insights in the analysis of biological data.


\section{Conclusions}
\label{sec:concl}

In this paper we presented probability measures as generalized functions, acting on appropriate function spaces. In this way we were able to introduce a new family of distances for probability measures, based on the evaluation of the PM functionals  on a finite number of well chosen functions of the sample. The calculation of these PM distances does not require the use of either parametric assumptions or explicit probability estimations which makes a clear advantage over most well established PM distances, such as Bregman divergences, which makes a clear advantage over most well established PM distances and divergences.

A battery of artificial and real data experiments have been used to study the performance of the new family of distances. Using synthetically generated data, we have shown their performance in the task of discriminating normally distributed data. Although the generated data sets are ideal for the use of T-statistics the new distances shows superior discrimination results than classical methods. Similar conclusions have been obtained when the proposed distances are used in homogeneity test scenarios. Regarding the practical applications, the new PM distances have been proven to be competitive in shape recognition problems. Also they represent a novel way to identify genes and discriminate between groups of patients in micro-arrays. 

In the near future we will afford the study of the geometry induced by the proposed measure and its asymptotic properties. It is also interesting to investigate in the relationship there may be exist between the proposed distance and other probability metrics.\\

\noindent
\textbf{Acknowledgements}
This work was partially supported by projects MIC 2012/00084/00, ECO2012-38442, SEJ2007-64500, MTM2012-36163-C06-06, DGUCM 2008/00058/002 and MEC 2007/04438/001.

\section*{Appendix}
\section*{A) Proof for Theorems of Section 3}

\begin{proof}[\textbf{of Theorem 1}] Consider $x \in \operatorname{Supp}(\mathbb{P})$; given $m$ and a sequence $\boldsymbol \alpha_\mathbb{P}^m$,  
$x \in A_i(\mathbb{P}) = S_{\alpha_i}(f_\mathbb{P}) - S_{\alpha_{i+1}}(f_\mathbb{P})$ for one (and only one) $i \in \{1,\dots,m-1\}$, that is $\alpha_i\leq f_\mathbb{P}(x) \leq \alpha_{i+1}$. Then $\phi_{i,\mathbb{P}}(x)=\indicator{A_i(\mathbb{P})}(x)=1$ in the region $A_i(\mathbb{P})$ and zero elsewhere. Given $\varepsilon > 0$, choose $m>\frac{1}{\varepsilon}$ and $\alpha_{i+1}= \alpha_i + \frac{1}{m}$. Given that $\alpha_i\leq f_\mathbb{P}(x) \leq \alpha_{i+1}$, then $|\alpha_{i} - f_\mathbb{P}(x)|\leq \frac{1}{m}$, and thus:

\begin{eqnarray}
|f_{m-1}(x) - f_\mathbb{P}(x) |  &= & \left|\sum_{j=1}^{m-1}\alpha_{j}\phi_{j,\mathbb{P}}(x) - f_\mathbb{P}(x)\right|\\ \nonumber
&= & |\alpha_i - f_\mathbb{P}(x)| \leq \frac{1}{m}< \varepsilon .
\end{eqnarray}
That is $\lim_{m\rightarrow \infty} f_{m-1}(x)= f_\mathbb{P}(x)$ pointwise  and also uniformly by Dini's Theorem. Therefore we can approximate (by fixing $m\gg 0$) the density function as a simple function, made up of linear combination of indicator functions weighted by coefficients that represents the density value of the $\alpha$-level sets of the density at hand.
\end{proof}

\begin{proof}[\textbf{of Corollary 1}]
By Theorem \ref{th1} we can approximate (by fixing $m \gg 0$) the density function as: 
$$f_\mathbb{P}(x) \approx \sum_{j=1}^{m-1}\alpha_{j}\phi_{j}(x),$$ 
where $\alpha_j= \langle \phi_{j},f_\mathbb{P}\rangle$ and $\phi_{j}$ is the indicator function of the $\alpha_j$-level set of $\mathbb{P}$. Then if $\langle \phi_{j},f_\mathbb{P}\rangle \xrightarrow[m \rightarrow \infty]{} \langle \phi_{j},f_\mathbb{Q}\rangle $, for all the indicator functions $\phi_{j}$, then $f_\mathbb{P} = f_\mathbb{Q}$.
\end{proof}

\begin{proof}[\textbf{of Proposition 1}]
If $\lim\limits_{m \rightarrow \infty} d_{{\boldsymbol{\alpha,\beta}}}(\mathbb{P},\mathbb{Q})=0$, then 
$$A_i(\mathbb{P}) \stackrel{m \rightarrow \infty}{=} A_i(\mathbb{Q}) \; \forall i$$. Thus $ f_{m}^{\mathbb{P}}(x) = \sum_{j=1}^{m}\alpha_{j}\phi_{j}(x) = f_{m}^{\mathbb{Q}}(x) \; \forall m$, and by Theorem \ref{th1}: $f_{\mathbb{P}} = f_{\mathbb{Q}}$. In the other way around if $\mathbb{P}=\mathbb{Q}$ ($f_{\mathbb{P}} = f_{\mathbb{Q}}$) then it is certain that $\lim\limits_{m \rightarrow \infty} d_{{\boldsymbol{\alpha,\beta}}}(\mathbb{P},\mathbb{Q})=0$.
\end{proof}

\section*{B) Invariance under affine transformations}
\begin{lemma}
Lebesgue measure is equivalent under affine transformations.
\end{lemma}

\begin{proof}
Let $X$ be a random variable that take values in $\mathbb{R}^d$ distributed according to $\mathbb{P}$, and let $f_\mathbb{P}$ be its density function. Let $T:\mathbb{R}^d \rightarrow \mathbb{R}^d$ be an affine transformation, define the r.v. $X^*=T(X)=a +bRX$, where $a\in \mathbb{R}^d$, $b\in \mathbb{R}^+$ and $R \in \mathbb{R}^{d\times d}$ is an orthogonal matrix with $\det(R)=1$ (therefore $R^{-1}$ exist and $R^{-1}=R^T$). Then $X^*$ is distributed according to $f_\mathbb{P^*}$. Define $E^*=\{x^*| x^*=T(x) \text{ and } x \in E\}$, then: 
\begin{eqnarray}
\mu^*(E^*)&= &\int_{E^*}d\mathbb{P^*}\\ \nonumber
&= &\int_{E^*} f_\mathbb{P^*}(x^*)dx^*\\ \nonumber
&= &\int_{E^*} f_\mathbb{P}(T^{-1}(x^*))\left| \frac{\partial T^{-1}(x^*)}{\partial x^*}\right| dx^*\\ \nonumber
& = & \int_{E^*} f_\mathbb{P}\left(R^{-1}\left(\frac{x^* - a}{b} \right) \right) \frac{R^{-1}}{b} dx^*\\ \nonumber
&= &\int_{E} f_\mathbb{P}(y)dy= \int_{E}d\mathbb{P}= \mu(E). \nonumber
\end{eqnarray}
\end{proof}

\begin{theorem}{\textbf{(Invariance under affine transformation)}}
The metric proposed in Eq. \eqref{eq:alpha-distance} is invariant under affine transformations.
\end{theorem}
\begin{proof}
Let $T$ be an affine transformation, we prove that the measure of the symmetric difference of any two $\alpha$-level sets is invariant under affine transformation, that is: $\mu(A_i(\mathbb{P})\triangle A_i(\mathbb{Q}))=\mu^*\left(T(A_i(\mathbb{P}))\triangle T(A_i(\mathbb{Q}))\right)=\mu^*(A_i(\mathbb{P}^*) \triangle A_i(\mathbb{Q}^*))$. By Lemma 1:
\begin{eqnarray} \nonumber
\mu^*(A_i(\mathbb{P}^*) \triangle A_i(\mathbb{Q}^*)) &= &    \int_{A_i(\mathbb{P}^*) - A_i(\mathbb{Q}^*)} d\mathbb{P}^* +   
        \int_{A_i(\mathbb{Q}^*)-A_i(\mathbb{P}^*)} d\mathbb{Q}^*\\ \nonumber
     &= &\int_{ A_i(\mathbb{P}) - A_i(\mathbb{Q})} d\mathbb{P} +   \int_{A_i(\mathbb{Q}) - A_i(\mathbb{P}) } d\mathbb{Q}\\ \nonumber
 & = & \mu(A_i(\mathbb{P})\triangle A_i(\mathbb{Q})). \nonumber
\end{eqnarray}
The same argument can be applied to the denominator in the expression given in Eq. \eqref{eq:alpha-distance}, thus 
$$\frac{w_i}{\mu^*(A_i(\mathbb{P}^*) \cup A_i(\mathbb{Q}^*))}=\frac{w_i}{\mu(A_i(\mathbb{P}) \cup A_i(\mathbb{Q}))}$$ 
for $i=1,\dots,m-1$. Therefore as this is true for all the $\alpha$-level sets, then the distance proposed in Eq. \eqref{eq:alpha-distance} is invariant under affine transformations:
\begin{eqnarray}\nonumber
d_{{\boldsymbol{\alpha,\beta}}}(\mathbb{P^*},\mathbb{Q^*}) &=& \sum_{i=1}^{m-1}w_i d\left(\phi_i(\mathbb{P^*}),\phi_i(\mathbb{Q^*})\right)\\\nonumber
&=&\sum_{i=1}^{m-1}\lambda_i
d\left(\phi_i(\mathbb{P}),\phi_i(\mathbb{Q})\right)\\ \nonumber
& = &d_{{\boldsymbol{\alpha,\beta}}}(\mathbb{P},\mathbb{Q}).
\end{eqnarray}
\end{proof}

 \section*{C) Details about the estimation of level sets}
  \begin{definition}[Neighbourhood Measures] Consider a random variable $X$ with density function $f(x)$ defined on $\bbbr^d$. Let $S_n$ denote the set of random independent identically distributed samples of size $n$ (drawn from $f$).
 The elements of $S_n$ take the form $s_n = (x_1,\cdots,x_n)$, where $x_i \in \bbbr^d$.
 Let $M: \bbbr^d \times S_n \longrightarrow \bbbr$ be a real-valued function defined for all $n \in \bbbn$.
 (a) If $f(x) < f(y)$ implies $\displaystyle\lim_{n\to\infty}P(M(x,s_n) > M(y,s_n)) = 1$, then $M$ is a \textbf{sparsity measure}. (b)
 If $f(x) < f(y)$ implies $\displaystyle\lim_{n\to\infty}P(M(x,s_n) < M(y,s_n)) = 1$, then $M$ is a \textbf{concentration measure}.
 \end{definition}
   
 \begin{example}
 Consider the distance from a point $x \in \bbbr^d$ to its $k^{th}$-nearest neighbour in $s_n$, $x^{(k)}$: $M(x,s_n) = d_k(x,s_n) = d(x,x^{(k)})$:
 it is a sparsity measure.
 \end{example}
 
 \begin{theorem}\label{th3}The set $R_n = \{x : h_n(x) = sign(\rho_n^*-g_n(x)) \ge 0\}$ converges to a region of the
 form $S_{\alpha}(f) = \{ x | f(x) \ge \alpha\}$, such that $P(S_{\alpha}(f)) = 1 - \nu$.
 Therefore,  the Support Neighbour Machine estimates a density contour cluster $S_{\alpha}(f)$ (around the mode).
 \end{theorem}
 
 Theorem \ref{th3} (see details in
  \citep{bib:munoz2006,bib:naiveoneclass}) ensures the convergence of the empirical estimation of the proposed distance. When the sample size increases, we are able to determine with more precision the sets $ A_i(\mathbb{P})$ and $ A_i(\mathbb{Q})$ and therefore $\hat d_{\boldsymbol{\alpha,\beta}}(\mathbb{P},\mathbb{Q})\xrightarrow{n} d_{\boldsymbol{\alpha,\beta}}(\mathbb{P},\mathbb{Q})$.
 
\section*{D) Extensions of the weighting scheme}
We present in this appendix alternative weighting schemes to LS(1):
\begin{enumerate}
\item Radius of the set $A_i(\mathbb{P})\triangle A_i(\mathbb{Q})$: Choose $w_i$ in Eq. \eqref{eq:alpha-distance} by:
\begin{equation}
\nonumber
w_i=\frac{1}{m}\max\limits_{x\in s_{\hat A_i(\mathbb{P})}, y\in s_{\hat A_i(\mathbb{Q})}} \left\{( 1- I_{r_{\hat A_i(\mathbb{P})},r_{\hat A_i(\mathbb{Q})}}(x,y))
\parallel x - y\parallel_2\right\}.
\end{equation}
\item Hausdorff distance between the sets 
$$A_i(\mathbb{P}) - A_i(\mathbb{Q})$$ 
and 
$$A_i(\mathbb{Q}) - A_i(\mathbb{P}):$$ 
Choose $w_i$ in Eq. \eqref{eq:alpha-distance} by:
\begin{equation}
\nonumber
w_i=\frac{1}{m}\hat H\left(s_{\hat A_i(\mathbb{Q})} - \hat A_i(\mathbb{P}), s_{\hat A_i(\mathbb{P})}- \hat A_i(\mathbb{Q}) \right),
\end{equation}
where $\hat H(\hat X,\hat Y)$ denotes the Hausdorff distance (finite size version) between finite sets $\hat X$ and $\hat Y$ (which estimates the `theoretical' Hausdorff distance between space regions $X$ and $Y$). In this case $X = A_i(\mathbb{P}) - A_i(\mathbb{Q})$ and
$Y = A_i(\mathbb{Q}) - A_i(\mathbb{P})$.
\end{enumerate}


\end{document}